\documentclass[conference]{IEEEtran}
\IEEEoverridecommandlockouts
\usepackage{cite}
\usepackage{amsmath,amssymb,amsfonts}
\usepackage{algorithm}
\usepackage{amsthm}
\usepackage{algorithmic}
\usepackage{graphicx}
\usepackage{textcomp}
\usepackage[table]{xcolor}
\usepackage{booktabs} % 表格美化
\usepackage{threeparttable}
\usepackage{tabularx}
\usepackage{multirow}
\usepackage{pgfplots}
\usepackage{adjustbox}
\usepackage[table]{xcolor}  % 或者 \usepackage{colortbl}
\pgfplotsset{compat=1.17}
\usepgfplotslibrary{groupplots}
\newtheorem{theorem}{Theorem}
\newtheorem{proposition}{Proposition}

\newcommand{\NA}{\textemdash}
\newcolumntype{Y}{>{\centering\arraybackslash}X}
\usepackage[letterpaper, top=0.75in, bottom=1in, left=0.625in, right=0.625in]{geometry}
\begin{document}

\title{HintKD: Hint Knowledge Distillation for Bandwidth-Constrained Cloud-Edge Inference}

\author{Wanling~Luo, Jiayi~Zhang,~\IEEEmembership{Senior Member,~IEEE}, Bokai~Xu, and Ziheng Liu
\thanks{W. Luo, J. Zhang, B. Xu, and Z. Liu are with State Key Laboratory of Advanced Rail Autonomous Operation, and also with School of Electronic and Information Engineering, Beijing Jiaotong University. (Corresponding authors: Jiayi Zhang; Bokai Xu.)}
}
\IEEEaftertitletext{\vspace{-1.5\baselineskip}}
%\IEEEaftertitletext{\vspace{-2\baselineskip}}  % 压缩标题块与后面内容的空白
\maketitle

\begin{abstract}

Collaborative inference between the edge cloud and user equipment (UE) is a promising paradigm for deploying large deep neural networks (DNNs) in 6G networks. However, existing cloud-edge inference and distillation schemes often require the real-time transmission of high-dimensional intermediate features or soft probability vectors, which imposes a substantial communication burden on bandwidth-limited wireless links. To address this challenge, we propose HintKD, a bandwidth-constrained distillation framework for cloud-edge inference. The core idea is to compress the teacher’s guidance into compact discrete hints rather than transmit raw features directly. Specifically, a cloud-side teacher is first trained with a maximum conditional mutual information (MCMI) objective to preserve informative intra-class variations. Its latent representation is then mapped to a compact codebook through differentiable vector quantization. During deployment, the cloud transmits only a hint index, and the UE uses a lightweight FiLM-based adapter to modulate student features according to the received codeword. Experiments on the DeepSense 6G beam prediction task show that HintKD preserves competitive accuracy while reducing the inference-side payload from high-dimensional feature or logit transmission by $\mathcal{O}(D_t) / \mathcal{O}(1)$ or $\mathcal{O}(C)/\mathcal{O}(1)$ times. In addition, experiments on the MNIST dataset further verify the generality of the proposed distillation method across different datasets and tasks. These results show that HintKD achieves a favorable accuracy-bandwidth trade-off by replacing high-dimensional feature or logit transmission with a compact discrete hint index.
\end{abstract}

% \begin{IEEEkeywords}
% 6G edge intelligence, knowledge distillation, bandwidth-constrained inference, vector quantization, mutual information
% \end{IEEEkeywords}

\section{Introduction}
\vspace{-0.2cm}
\begin{figure*}[t]
    \centering
    \includegraphics[width=0.6\linewidth]{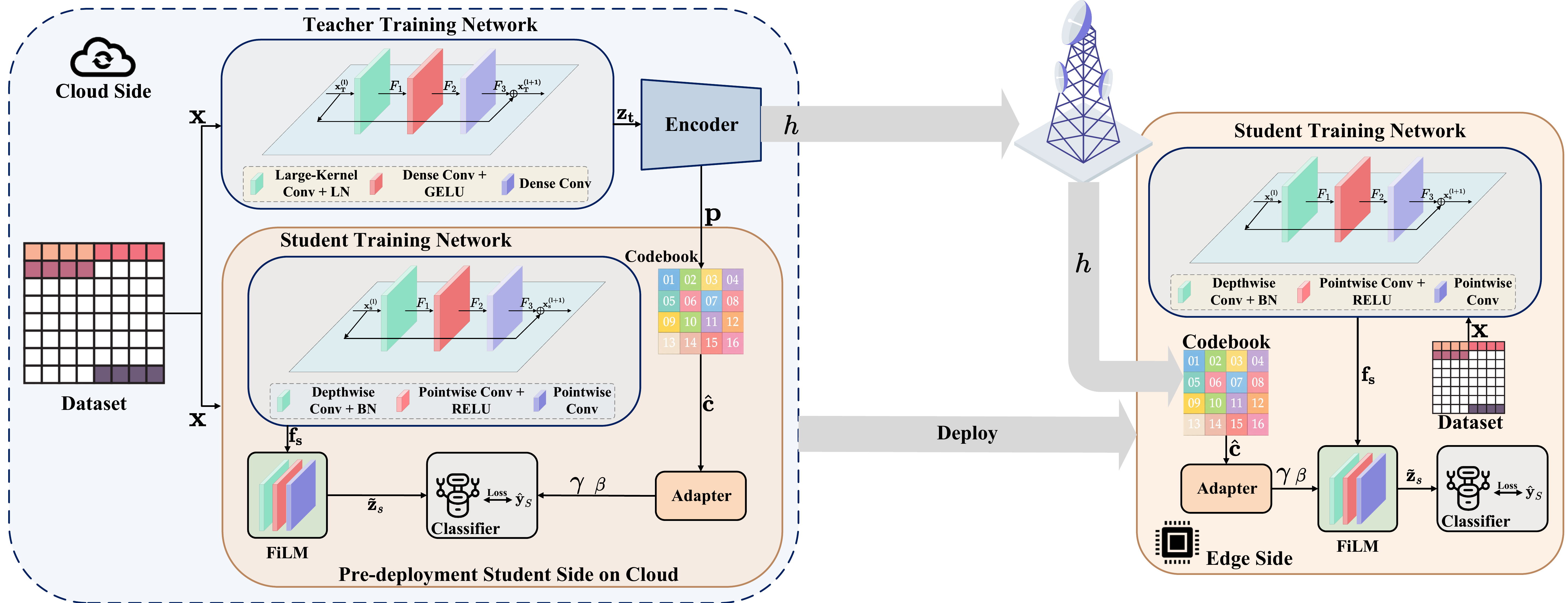} % 替换为你的 MNIST Top-5 图片文件名
    \caption{Overall architecture of the proposed HintKD framework.}
    \label{fig:figure}
\end{figure*}
With the rapid proliferation of the Internet of Things (IoT) and the advent of 6G networks, deploying artificial intelligence (AI) at the network edge has become increasingly important\cite{letaief2019roadmap, zhou2019edge, 11058968}. However, user equipment (UE) is typically constrained by limited battery capacity, memory, and computing resources, which makes it difficult to deploy large-scale deep neural networks (DNNs) directly on end devices\cite{letaief2019roadmap, zhou2019edge}. To address this challenge, existing research has mainly followed two directions: cloud-edge collaboration and on-device intelligence through model compression. Recent 6G network designs increasingly integrate edge-device collaboration with native artificial intelligence capabilities \cite{wang2025edgedevice}.

Cloud-edge collaboration reduces the computational burden on UEs by allowing resource-intensive processing to be supported by the cloud or edge server \cite{matsubara2022split, eshratifar2019bottlenet}. In this paradigm, the UE performs only lightweight local computation, while the remaining inference process relies on cloud-side assistance. Although such collaboration can effectively exploit powerful infrastructure, it also creates a strong dependence on wireless connectivity. In particular, when intermediate information must be exchanged during inference, system performance becomes highly sensitive to round-trip latency, channel quality, and uplink bandwidth \cite{xie2021deep, bourtsoulatze2019deep, matsubara2022split, eshratifar2019bottlenet}. As a result, the quality of service degrade sharply in poor channel conditions or coverage holes. At the physical layer, recent studies have investigated fluid-antenna-assisted integrated sensing and communications for unsourced random access as promising mechanisms for exploiting spatial diversity, supporting massive connectivity in future wireless networks \cite{zhang2026fluidisac,zhang2025slowfama}.

Compared with pure cloud-assisted inference, knowledge distillation (KD) offers a more attractive way to endow UEs with autonomous inference capability. In standard KD, a lightweight student model learns from a stronger teacher model by matching its output distribution \cite{hinton2015distilling}. Subsequent studies extended this idea from logits to richer intermediate representations, including feature hints, attention maps, and correlation-based relations, in order to improve the student’s generalization ability \cite{romero2015fitnets, zagoruyko2017paying, peng2019correlation}. These advances have made KD an effective framework for transferring knowledge from large models to compact edge models.

However, most existing KD methods are developed under offline or local settings, where the teacher and student are co-located or where communication is not a limiting factor. In a realistic wireless edge architecture, by contrast, the teacher is naturally deployed on the cloud or base station side, while the student runs on the UE. This physical separation fundamentally changes the distillation problem: if the teacher’s guidance is transmitted as high-dimensional features or soft logits, the communication overhead can become prohibitive for bandwidth-constrained wireless links \cite{tian2020contrastive}. In particular, feature-based transmission scales with the teacher feature dimension, and logits-based transmission scales with the number of classes, making both approaches inefficient for real-time cloud-edge inference. Related studies on deep CSI compression and feedback further demonstrate the importance of reducing the dimensionality of transmitted wireless representations in bandwidth-limited 6G systems \cite{lu2026csicompression}.

To address this problem, we propose HintKD, a bandwidth-constrained distillation framework for cloud-edge inference. Instead of transmitting raw teacher features or full logit vectors, HintKD compresses the teacher’s guidance into a compact discrete hint. More specifically, a cloud-side teacher is first trained with a maximum conditional mutual information (MCMI) objective so that its representation preserves informative intra-class variation. The resulting latent feature is then mapped into a compact codebook through differentiable vector quantization. During deployment, the cloud transmits only the hint index, and the UE uses a lightweight FiLM-based adaptation module to modulate local student features according to the received codeword. In this way, remote teacher guidance is converted into a low-bit control signal suitable for bandwidth-limited edge intelligence.

The main contributions of this paper are summarized as follows. First, we develop a cloud-edge distillation framework that replaces high-dimensional teacher transmission with compact discrete hint signaling. Second, we integrate MCMI-based teacher learning, differentiable codebook design, and FiLM-based student adaptation into a unified architecture for bandwidth-constrained inference. Third, we analyze the communication complexity of the proposed framework and show that the inference-side payload is reduced from feature- or logit-level real-valued transmission to a low-bit discrete index. Finally, experiments on the DeepSense 6G beam prediction task demonstrate that HintKD preserves competitive predictive accuracy while achieving a more favorable accuracy-bandwidth trade-off than conventional Feature KD and Logits KD baselines, and additional experiments on MNIST further verify the generality of the proposed method across different datasets and tasks.

\vspace{-0.2cm}
\section{HintKD Cloud-Edge Inference Framework}
\label{sec:method}
\vspace{-0.2cm}
In this section, we elaborate on the proposed HintKD framework, as illustrated in the system architecture Fig. \ref{fig:figure}.
\vspace{-0.2cm}
\subsection{Training Formulation}
\vspace{-0.2cm}
Consider a generic classification task with a dataset $\mathcal{D} = \{(\mathbf{x}_i, y_i)\}_{i=1}^{N}$, where $\mathbf{x}_i \in \mathbb{R}^{D_{in}}$ represents the $i$-th input data vector with dimension $D_{in}$, and $y_i \in \{1, \dots, C\}$ denotes its corresponding scalar class label among $C$ distinct classes. We deploy a high-capacity Teacher model $T$ on the cloud and a lightweight Student model $S$ on the edge device.
We define the Teacher model as a composite function $T(\cdot) = g_T(f_T(\cdot))$, where:
$f_T: \mathbb{R}^{D_{in}} \to \mathbb{R}^{D_t}$ is the backbone that maps input $\mathbf{x}$ to a latent feature vector $\mathbf{z}_t$. Thus, for any input $\mathbf{x}$, we explicitly define the teacher feature as $\mathbf{z}_t = f_T(\mathbf{x}; \theta_T)$. And $g_T: \mathbb{R}^{D_t} \to \mathbb{R}^C$ is the classifier head that maps $\mathbf{z}_t$ to the probability logits.
The final prediction is $\hat{\mathbf{p}}_T = \text{Softmax}(g_T(\mathbf{z}_t))$.
To ensure the teacher captures comprehensive contextual information rather than merely fitting the hard labels, we adopt the maximum conditional mutual information (MCMI) principle \cite{ye2024bayes}. The training objective for the teacher is formulated as:
\vspace{-0.2cm}
\begin{equation}
\mathcal{L}_{T} = \frac{1}{B}\sum_{i=1}^{B} \bigl(-\log \hat{p}_{T,i}(y_i)\bigr) - \lambda I(X; \hat{Y}_T | Y),
\label{eq:teacher_loss}
\end{equation}
where \(\hat{p}_{T,i}(y_i)\) denotes the probability assigned by the teacher to the ground-truth label \(y_i\). The hyperparameter $\lambda > 0$ controls the trade-off between classification accuracy and information maximization. For the information-theoretic term, $X$, $Y$, and $\hat{Y}_T$ denote the random variables corresponding to the high-dimensional input, class labels, and prediction space, respectively. Maximizing the conditional mutual information $I(X; \hat{Y}_T | Y)$ encourages the teacher's prediction to contain more information about the input $X$ given the class $Y$ and forces the Teacher to retain intra-class environmental features within its output distribution, serving as a rich source for subsequent hint generation.  

Following \cite{ye2024bayes}, we approximate the intractable CMI term using its empirical estimator. For a mini-batch of size $B$, this is computed as the average Kullback-Leibler (KL) divergence between the teacher's output for an individual sample $\mathbf{x}_i$ and the class-conditional centroid $\mathbf{q}_{y_i}$:
\vspace{-0.2cm}
\begin{equation}
I(X; \hat{Y}_T | Y) \approx \frac{1}{B} \sum_{i=1}^{B} \text{KL}(\hat{\mathbf{p}}_T(\mathbf{x}_i) \| \mathbf{q}_{y_i}),
\label{eq:I}
\end{equation}

where $\mathbf{q}_{c} \in \mathbb{R}^C$ is the empirical class centroid for class $c$, calculated as $\mathbf{q}_{c} = \frac{1}{N_c}\sum_{j:\, y_j=c}\hat{\mathbf{p}}_T(\mathbf{x}_j), N_c = \bigl|\{j: y_j=c\}\bigr|$. By maximizing this divergence, the teacher is forced to learn diverse, sample-specific ``environmental features" distinct from the class prototype, providing richer guidance for the subsequent hint generation.

We first optimize the teacher's parameters $\boldsymbol{\theta}_T$ using the MCMI objective defined in (1). Let $f_T(\cdot; \boldsymbol{\theta}_T)$ denote the feature extraction part of the trained teacher. For any input $\mathbf{x}$, it produces a high-dimensional latent feature vector $\mathbf{z}_t = f_T(\mathbf{x}; \boldsymbol{\theta}_T) \in \mathbb{R}^{D_t}$, which serves as the rich source for the subsequent hint generation.

Subsequently, we freeze $\boldsymbol{\theta}_T$. The goal is to train the student to predict $\mathbf{y}$ given the input $\mathbf{x}$, but conditioned on a quantized version of the teacher's feature $\mathbf{z}_t$. The framework then jointly optimizes the Student's parameters $\boldsymbol{\theta}_S$ and the Hint mechanism's parameters $\boldsymbol{\psi}$. The goal is to minimize the prediction error of the student $S(\cdot; \boldsymbol{\theta}_S)$, conditioned on the discrete hints derived from the frozen teacher's feature $\mathbf{z}_t$. This is formulated as:
\vspace{-0.2cm}
\begin{equation}
\min_{\boldsymbol{\theta}_S, \boldsymbol{\psi}} \mathbb{E}_{(\mathbf{x}, \mathbf{y}) \sim \mathcal{D}} \left[ \mathcal{L}_{CE}\Big(\mathbf{y}, \hat{\mathbf{y}}_S\big(\mathbf{x} \mid \mathcal{Q}(\mathbf{z}_t; \boldsymbol{\psi}); \boldsymbol{\theta}_S\big)\Big) \right],
\label{eq:student_obj_general}
\end{equation}
where $\mathcal{Q}(\cdot; \boldsymbol{\psi})$ represents the quantization operation mapping the continuous $\mathbf{z}_t$ to discrete indices, $\mathbf{y} \in \left \{  0, 1 \right \}^C $ denotes the one-hot ground truth vector and $\hat{\mathbf{y}}_S \in  \mathbb{R}^C$ represents the student's predicted probability density vector. This formulation explicitly bridges the fixed teacher's knowledge with the student's inference logic\cite{xu2023resource}.
\vspace{-0.2cm}
\subsection{Cloud-Side: Contextual Vector Quantization}
\vspace{-0.2cm}
This subsection details the vector quantization mechanism $\mathcal{Q}(\mathbf{z}_t; \boldsymbol{\psi})$ introduced in (\ref{eq:student_obj_general}) designed for the cloud-side teacher. It takes the continuous teacher feature $\mathbf{z}_t$ as input and outputs a discrete index. Although jointly optimized with the student during the offline training phase, this module is specifically engineered to compress the teacher's high-dimensional features into compact codes for bandwidth-constrained transmission during deployment.
To enable bandwidth-limited transmission, we map $\mathbf{z}_t$ into a discrete codebook $\mathcal{C} = \{\mathbf{e}_k\}_{k=1}^K$, where $\mathbf{e}_k \in \mathbb{R}^{D_h}$ represents the $k$-th latent prototype vector, and $K$ is the codebook size, inspired by vector quantized variational autoencoders (VQ-VAE)\cite{van2017neural}.

During training, we employ the Gumbel-Softmax relaxation to bridge the gradient flow. The soft probability vector $\mathbf{p} = [p_1, \dots, p_K]^\top \in \Delta^{K-1}$ is computed as:
\vspace{-0.2cm}
\begin{equation}
\label{eq:gumbel}
p_k = \frac{\exp((\ell_k + g_k) / \tau)}{\sum_{j=1}^K \exp((\ell_j + g_j) / \tau)},
\end{equation}
where $\boldsymbol{\ell} = f_{\mathrm{enc}}(\mathbf{z}_t;\boldsymbol{\psi}_{e}) :\mathbb{R}^{D_t} \to \mathbb{R}^{D_h}$ is a projection mapping that aligns the teacher feature dimension with the codebook, $g_k \sim \text{Gumbel}(0, 1)$ represents independent noise samples, and $\tau$ is the temperature parameter.

\begin{algorithm}[t]
\caption{Training and Inference Procedure of the HintKD Framework}
\small
\label{alg:hintkd}
\begin{algorithmic}[1]
\REQUIRE Dataset $\mathcal{D}$; teacher feature extractor $f_T$ and head $g_T$; Student Backbone $S$, Encoder $f_{enc}$, Codebook $\mathcal{C} = \{\mathbf{e}_k\}_{k=1}^{K}$, Adapter $(\Phi_\gamma, \Phi_\beta)$, Temperature $\tau$, MCMI weight $\lambda$.

\STATE \textbf{Phase 1: Teacher Pre-training}
\WHILE{not converged}
    \STATE Sample mini-batch $B = \left \{ (\mathbf{x_i}, y_i) \right \}^B_{i=1}s  \sim \mathcal{D}$ 
    \STATE Compute $\mathbf{z}_T = f_T(\mathbf{x})$, and $\hat{\mathbf{p}}_T = $ softmax $(g_T(\mathbf{z}_T))$
    \STATE Estimate $I(X; \hat{Y}_T | Y)$ using (\ref{eq:teacher_loss})
    \STATE Update $\theta_T$ by minimizing $\mathcal{L}_T$ in (\ref{eq:I})
\ENDWHILE
\STATE \textbf{Freeze} $\theta_T$.

\STATE \textbf{Phase 2: bandwidth-constrained distillation}
\WHILE{not converged}
    \STATE Sample mini-batch $B = \left \{ (\mathbf{x_i}, y_i) \right \}^B_{i=1}  \sim \mathcal{D}$
    
    \STATE Compute frozen teacher feature: $\mathbf{z}_t = f_T(\mathbf{x})$
    \STATE Compute encoder logits $ \boldsymbol{\ell} = f_{enc}(\mathbf{z}_t)$
    \STATE Obtain relaxed assignment $\mathbf{p}$ using (\ref{eq:gumbel})
    \STATE Modulate feature in (\ref{eq:film_compound})
    \STATE Compute student prediction $\hat{\mathbf{p}}_S$ = Softmax $(g_s(\tilde{z}_s))$
    \STATE Update $\theta_S, \psi_e, \mathcal{C}$, and $\phi$ by minimizing $\mathcal{L}_S$ in (\ref{eq:Loss_S})
\ENDWHILE

\STATE \textbf{Phase 3: Deployment}
\STATE \textbf{Cloud:} Compute $\boldsymbol{\ell} = f_{enc}(f_T(\mathbf{x}))$ and $h = \arg\max_k(\boldsymbol{\ell}_k))$
\STATE \textbf{Transmission:} Send only index $h$ ($\lceil \log_2 K \rceil$ bits)
\STATE \textbf{Edge:} Retrieve $\mathbf{e}_{h}$ and produce the final prediction
\end{algorithmic}
\end{algorithm}

During deployment, we perform hard quantization by selecting the index of the mode, i.e., $h = \operatorname*{arg\,max}_k \boldsymbol{\ell}_k$. This ensures that the transmission payload is strictly bounded by the codebook entropy $\lceil \log_2 K \rceil$ bits, achieving a constant communication overhead $\mathcal{O}(1)$ regardless of the feature dimension $D_t$. In stark contrast, traditional Feature KD requires transmitting the raw latent vector with a complexity of $\mathcal{O}(D_t)$, while Logits KD scales with the number of classes as $\mathcal{O}(C)$. Since typically 
$\mathcal{O}(1) < \mathcal{O}(C) < \mathcal{O}(D_t)$
, our method establishes an exponential reduction in transmission payload.

It is crucial to distinguish the transmission payload between phases. 
During the offline joint training, to preserve differentiability for back-propagation, the Cloud-side transmits the soft probability vector $\mathbf{p} \in \mathbb{R}^K$ (4) rather than the discrete index. Even in this high-precision mode, our method remains efficient. Since the codebook size $K$ is typically smaller than the feature dimension $D_t$ used in Feature KD or the class number $C$ in Logits KD, HintKD reduces the training communication overhead by $\frac{D_t}{K} \times$ and $\frac{C}{K} \times$, respectively. 
During online inference, the transmission switches to the hard index $h$, achieving the ultimate bandwidth compression.
\vspace{-0.2cm}
\subsection{Edge-Side: Hint Adaptation via FiLM}
\vspace{-0.2cm}
This subsection describes the lightweight adaptation mechanism deployed on the UE. Its role is to interpret the received hint indices and modulate the student network's inference logic locally.

The UE receives the hint index $h$ and retrieves the corresponding vector $\mathbf{e}_h$ from the local codebook. To integrate this global guidance into the local inference, we formulate the student's processing as a conditionally modulated transformation \cite{perez2017film}. We unify the feature extraction and parameter generation into a single composite function. Let $\mathbf{x} \in \mathbb{R}^{C_{in} \times H \times W}$ denote the raw channel input. The hint-modulated representation $\tilde{\mathbf{z}}_s \in \mathbb{R}^{D_s}$ is defined as:
\vspace{-0.2cm}
\begin{equation}
\tilde{\mathbf{z}}_s(\mathbf{x}, \mathbf{e}_h) = \underbrace{(\mathbf{1} + \Phi_\gamma(\mathbf{e}_h; \phi))}_{\text{Scale Term}} \odot \underbrace{S(\mathbf{x}; \theta_S)}_{\text{Raw Feature}} + \underbrace{\Phi_\beta(\mathbf{e}_h; \phi)}_{\text{Shift Term}},
\label{eq:film_compound}
\end{equation}
where $S(\cdot; \theta_S)$ is the student backbone network parameterized by weights $\theta_S$, which maps the raw input $\mathbf{x}$ to the local feature space $\mathbb{R}^{D_s}$. And $\Phi_\gamma, \Phi_\beta(\cdot; \phi)$ are the adapter projection heads parameterized by weights $\phi$, which map the discrete hint $\mathbf{e}_h$ to the modulation parameters $\boldsymbol{\gamma}, \boldsymbol{\beta} \in \mathbb{R}^{D_s}$. Crucially, (\ref{eq:film_compound}) preserves the dimensionality ($D_s$), allowing seamless integration with the classifier.
Finally, the modulated feature vector $\tilde{\mathbf{z}}_s \in \mathbb{R}^{D_s}$ is fed into the student's task-specific classifier head to obtain the final prediction probability vector $\hat{\mathbf{p}}_S \in \mathbb{R}^C$ and $\mathcal{L}_S$: 
\vspace{-0.2cm}
\begin{equation}
\hat{\mathbf{p}}_S = \text{Softmax}(g_s(\tilde{\mathbf{z}}_s)),
\label{eq:final_prediction}
\end{equation}
\vspace{-0.2cm}
\begin{equation}
    \mathcal{L}_{S}
=
\frac{1}{B}\sum_{i=1}^{B} \bigl(-\log \hat{p}_{S,i}(y_i)\bigr),
\label{eq:Loss_S}
\end{equation}
\noindent This prediction $\hat{\mathbf{p}}_S$ is then used to compute the cross-entropy loss against the ground truth labels. While $\hat{p}_S,i(y_i)$ denotes the probability assigned by the student. (\ref{eq:Loss_S}) clarifies the complete inference chain from the transmitted hint to the student's final decision.

The training procedure is summarized in Algorithm \ref{alg:hintkd}.
\vspace{-0.2cm}
\subsection{Theoretical Analysis}
\vspace{-0.2cm}
In this subsection, we analyze the generalization capability and communication complexity of the proposed framework.

\begin{theorem}[Hint-conditioned generalization bound]
\label{thm:generalization}
Let $\mathcal{H}$ be the hypothesis space of the student model. By conditioning on the discrete hints $\mathcal{C} = \{1, \dots, K\}$, the HintKD framework effectively partitions $\mathcal{H}$ into $K$ subspaces $\{\mathcal{H}_k\}_{k=1}^K$, where $\mathcal{H}_k$ corresponds to the student's function mapping conditioned on hint index $k$. Let $\hat{\mathcal{R}}_S(\mathcal{H})$ denote the empirical Rademacher complexity and $N$ denotes the number of training samples. The generalization error of the student, $\epsilon_{s}$, is bounded with probability $1-\delta$ by:
\vspace{-0.2cm}
\begin{equation}
\label{eq:rademacher_bound}
\epsilon_{s} \!\leq\! \hat{\epsilon}_{emp} \!+\! 2\sum_{k=1}^K \frac{n_k}{N} \hat{\mathcal{R}}_{S_k}(\mathcal{H}_k) \!+\! 3\sqrt{\frac{\log(2/\delta)}{2N}} \!+\! \Delta_{align},
\end{equation}
where $S_k$ is the subset of training samples assigned to hint $k$, $n_k = |S_k|$and $\Delta_{align}$ is the teacher-student alignment error.
\end{theorem}

\begin{IEEEproof}
(Sketch) The standard generalization bound is governed by the global Rademacher complexity: 
\vspace{-0.2cm}
\begin{equation}
  \hat{\mathcal{R}}_S(\mathcal{H}) = \mathbb{E}_{\boldsymbol{\sigma}}[\sup_{f \in \mathcal{H}} \frac{1}{N} \sum_{i=1}^N \sigma_i f_{h_i}(x_i)] \le \sum_{k=1}^K \frac{n_k}{N}\hat{\mathcal{R}}_{S_k}(\mathcal{H}_k),
\end{equation}
In HintKD, the FiLM mechanism switches the student's operational mode based on the hint $h$. This decomposes the learning problem into $K$ conditioned sub-problems. Since the conditional distribution $P(Y|X, h=k)$ has lower entropy than the marginal $P(Y|X)$, the hypothesis subspace $\mathcal{H}_k$ required to fit the data subset $S_k$ is significantly less complex than the global space $\mathcal{H}$. Consequently, the weighted sum of local complexities $\sum p_k \hat{\mathcal{R}}_{S_k}(\mathcal{H}_k)$ is strictly strictly lower than the global complexity $\hat{\mathcal{R}}_S(\mathcal{H})$, leading to a tighter bound.
\end{IEEEproof}

\begin{proposition}[Per-sample communication complexity]
Assume that continuous values are transmitted in 32-bit floating-point format. For one sample, Feature KD transmits \(D_t\) real values, Logits KD transmits \(C\) real values, and HintKD transmits a \(K\)-dimensional soft assignment during training and one discrete index during inference. Hence
\vspace{-0.2cm}
\begin{equation}
\mathcal{C}_{\mathrm{Feat}}^{\mathrm{train}}
=
\mathcal{C}_{\mathrm{Feat}}^{\mathrm{infer}}
=
32D_t,
\qquad
\mathcal{C}_{\mathrm{Logits}}^{\mathrm{train}}
=
\mathcal{C}_{\mathrm{Logits}}^{\mathrm{infer}}
=
32C,
\label{eq:feat_logits_cost}
\end{equation}
\begin{equation}
\mathcal{C}_{\mathrm{Hint}}^{\mathrm{train}}
=
32K,
\qquad
\mathcal{C}_{\mathrm{Hint}}^{\mathrm{infer}}
=
\lceil \log_2 K \rceil .
\label{eq:hint_cost}
\end{equation}
Therefore, HintKD reduces the inference-side communication from real-valued payloads that scale linearly with \(D_t\) or \(C\) to an integer payload that scales only logarithmically with \(K\). In addition, HintKD is more communication-efficient than Feature KD during training whenever \(K<D_t\), and more efficient than Logits KD during training whenever \(K<C\).
\end{proposition}
\begin{IEEEproof} 
In typical Edge AI scenarios, the dimensions satisfy the inequality $K < C \ll D_t$. Consequently, our framework establishes a strictly lower bound on communication overhead across the entire lifecycle:
$    \mathcal{C}_{Hint}^{infer} \ll \mathcal{C}_{Hint}^{train} < \mathcal{C}_{Logits} \ll \mathcal{C}_{Feat}$
Specifically, during deployment, HintKD achieves an exponential reduction in payload from linear complexity $\mathcal{O}(K)$ to logarithmic complexity $\mathcal{O}(\log K)$, rendering it the only feasible solution for ultra-low-bandwidth control channels. \hfill
\end{IEEEproof}

\begin{table}[t]
\centering
\caption{THE MODEL PARAMETER SETTINGS IN OUR EXPERIMENTS.}
\label{tab:parameters}
\renewcommand{\arraystretch}{1.05} % 增加行间距，提升阅读体验
\begin{tabular}{lcc}
\toprule
\textbf{Description} & \textbf{Symbol} & \textbf{Value} \\
\midrule
\multirow{2}{*}{Evaluated datasets} & \multirow{2}{*}{-} & DeepSense 6G (Scen. 9) \\ & & MNIST \\
\multirow{2}{*}{Primary tasks} & \multirow{2}{*}{-} & Beam Prediction \\ & & Image Classification \\
\multirow{2}{*}{Input dimensions} & \multirow{2}{*}{$C \times H \times W$} & $4 \times 256 \times 128$ (Radar) \\
 & & $1 \times 28 \times 28$ (MNIST) \\
Number of classes & $N_{\text{cls}}$ & 33 (Radar), 10 (MNIST) \\
Teacher model backbone & - & LeNet \\
Student model backbone & - & Light-CNN \\
Teacher feature dimension & $D_t$ & 84 \\
Student feature dimension & $D_s$ & 84 \\
Hint bandwidth budget & $B$ & $2, 4$ bits \\
Codebook size & $K$ & $4, 16$ \\
Hint injection method & - & FiLM \\
Optimizer & - & Adam \\
Total training epochs & $E$ & 40 (Radar), 5 (MNIST) \\
Batch size & $B_{\text{size}}$ & 64 \\
Initial learning rate & $\eta$ & $1 \times 10^{-3}$ \\
Weight decay & $\lambda$ & $5 \times 10^{-4}$ \\
Gumbel-Softmax temperature & $\tau$ & $1.0$ \\
\bottomrule
\end{tabular}
\end{table}

\begin{table*}[t]
\centering
\caption{COMPREHENSIVE PERFORMANCE COMPARISON ON DEEPSENSE 6G AND MNIST DATASETS.}
\label{tab:main_results}
\scriptsize
\setlength{\tabcolsep}{3pt}
\renewcommand{\arraystretch}{1.05}
\begin{adjustbox}{max width=\textwidth}
\begin{threeparttable}
\begin{tabularx}{\textwidth}{@{} l *{3}{Y} !{\color{black!25}\vrule width 0.45pt} *{2}{Y} !{\color{black!25}\vrule width 0.45pt} *{2}{Y} @{}}
\toprule
\rowcolor{blue!7}[\dimexpr\tabcolsep+0.3pt\relax][0pt]
\textbf{Method} &
\multicolumn{3}{c!{\color{black!25}\vrule width 0.45pt}}{\textbf{Teacher Setup}} &
\multicolumn{2}{c!{\color{black!25}\vrule width 0.45pt}}{\textbf{Cost (bits)}} &
\multicolumn{2}{c}{\textbf{Accuracy (\%)}} \\
\rowcolor{blue!7}[\dimexpr\tabcolsep+0.3pt\relax][0pt]
& \textbf{Obj.} & \textbf{LL} & \textbf{CMI} & \textbf{Train} & \textbf{Infer.} & \textbf{Top-1} & \multicolumn{1}{>{\columncolor{blue!7}[\dimexpr\tabcolsep+0.3pt\relax][-\dimexpr\tabcolsep\relax]}c}{\textbf{Top-5}} \\
\midrule
\multicolumn{8}{@{}l}{\textbf{Scenario 9}\cite{alkhateeb2023deepsense}} \\
\midrule
Cloud & MLL & -1.432 & 0.195 & \NA & \NA & \textbf{55.45} & 97.49 \\
Cloud & MCMI & -1.248 & 0.363 & \NA & \NA & 51.65 & \textbf{98.43} \\
\addlinespace[1pt]
Offline & \NA & \NA & \NA & 0 & 0 & 47.12 & 96.93 \\
\addlinespace[1pt]
Feature KD \cite{romero2015fitnets} & MLL & -1.430 & 0.200 & 2688 & 2688 & 46.06 & 97.49 \\
Feature KD \cite{romero2015fitnets} & MCMI & -1.250 & 0.360 & 2688 & 2688 & 45.05 & 97.43 \\
Logits KD \cite{hinton2015distilling} & MLL & -1.430 & 0.200 & 1056 & 1056 & 48.02 & \textbf{97.71} \\
Logits KD \cite{hinton2015distilling} & MCMI & -1.250 & 0.360 & 1056 & 1056 & 47.50 & \textbf{97.71} \\
\rowcolor{black!5}[\dimexpr\tabcolsep+0.3pt\relax][\dimexpr\tabcolsep+0.3pt\relax]
HintKD & MLL & -1.430 & 0.200 & 512 & 4 & 46.95 & 97.43 \\
\rowcolor{black!5}[\dimexpr\tabcolsep+0.3pt\relax][\dimexpr\tabcolsep+0.3pt\relax]
HintKD & MCMI  & -1.250 & 0.360 & 512 & 4 & \textbf{49.53} & 97.65 \\
\rowcolor{black!5}[\dimexpr\tabcolsep+0.3pt\relax][\dimexpr\tabcolsep+0.3pt\relax]
HintKD & MCMI  & -1.250 & 0.360 & 128 & 2 & 45.44 & 96.76 \\
\rowcolor{black!5}[\dimexpr\tabcolsep+0.3pt\relax][\dimexpr\tabcolsep+0.3pt\relax]
HintKD & MCMI & -1.250 & 0.360 & 2048 & 6 & 47.07 & 97.26 \\
\midrule
\multicolumn{8}{@{}l}{\textbf{MNIST}\cite{lecun1998gradient}} \\
\midrule
Cloud & MLL & -0.010 & 0.090 & \NA & \NA & \textbf{99.62} & \textbf{100.00} \\
Cloud & MCMI & -0.070 & 0.820 & \NA & \NA & 99.14 & \textbf{100.00} \\
\addlinespace[1pt]
Offline & \NA & \NA & \NA & 0 & 0 & 99.05 & 99.85 \\
\addlinespace[1pt]
Feature KD \cite{romero2015fitnets} & MLL & -0.010 & 0.090 & 2688 & 2688 & 98.95 & \textbf{99.99} \\
Feature KD \cite{romero2015fitnets} & MCMI & -0.070 & 0.820 & 2688 & 2688 & 98.88 & \textbf{99.99} \\
Logits KD \cite{hinton2015distilling} & MLL & -0.010 & 0.090 & 320 & 320 & 98.95 & 99.98 \\
Logits KD \cite{hinton2015distilling} & MCMI & -0.070 & 0.820 & 320 & 320 & 98.36 & \textbf{99.99} \\
\rowcolor{black!5}[\dimexpr\tabcolsep+0.3pt\relax][\dimexpr\tabcolsep+0.3pt\relax]
HintKD & MLL & -0.010 & 0.090 & 512 & 4 & \textbf{99.17} & 99.97 \\
\rowcolor{black!5}[\dimexpr\tabcolsep+0.3pt\relax][\dimexpr\tabcolsep+0.3pt\relax]
HintKD & MCMI & -0.070 & 0.820 & 512 & 4 & 99.06 & 99.97 \\
\rowcolor{black!5}[\dimexpr\tabcolsep+0.3pt\relax][\dimexpr\tabcolsep+0.3pt\relax]
HintKD & MCMI  & -0.070 & 0.820 & 128 & 2 & 99.04 & 99.97 \\
\rowcolor{black!5}[\dimexpr\tabcolsep+0.3pt\relax][\dimexpr\tabcolsep+0.3pt\relax]
HintKD & MCMI & -0.070 & 0.820 & 2048 & 6 & 99.14 & 99.97 \\
\bottomrule
\end{tabularx}
\end{threeparttable}
\end{adjustbox}
\end{table*}
\vspace{-0.2cm}
\section{Simulation Results}
\vspace{-0.2cm}
% Fig.~\ref{fig:mnist_epochs} reports the convergence behavior on MNIST. Since MNIST is much easier than DeepSense beam prediction, all variants quickly approach saturated accuracy, and the gap among different hint settings becomes small. Nevertheless, the proposed method remains stable under low-bit transmission, showing that the discrete-hint design does not degrade performance on simple visual classification tasks.

To evaluate the effectiveness of the proposed HintKD framework, we conduct experiments on both the DeepSense 6G dataset and the MNIST dataset. Among them, DeepSense 6G Scenario 9 is used as the main beam prediction task to verify the effectiveness of the proposed method under bandwidth-constrained cloud-edge inference, while MNIST is further introduced to validate the generality of the proposed distillation method across different data modalities and classification tasks\cite{alkhateeb2023deepsense, lecun1998gradient}. Specifically, Scenario 9 corresponds to a radar-based beam prediction problem, whereas MNIST is a standard handwritten digit classification benchmark. Therefore, the joint evaluation on these two datasets helps demonstrate that HintKD is not limited to a specific sensing modality, but can also maintain effectiveness on conventional image classification tasks. The detailed experimental settings are summarized in Table ~\ref{tab:parameters}, including the evaluated datasets, input dimensions, backbone configurations, hint bandwidth, codebook size, and training hyperparameters.

Fig.~\ref{fig:mll} compares different methods in terms of Accuracy, CMI, and LL. It can be observed that the teacher trained with the MCMI objective consistently achieves a higher CMI value, while the MLL-trained teacher exhibits a lower LL value. This indicates that the two teacher objectives emphasize different representation properties. More importantly, such differences are propagated to the student side and further affect the final distillation performance, showing that the teacher objective plays an important role in the proposed hint-based cloud-edge distillation framework.

Table~\ref{tab:main_results} reports the comprehensive results on DeepSense 6G Scenario 9. The offline student achieves 47.12\% Top-1 accuracy and 96.93\% Top-5 accuracy. Under the default 4-bit setting, HintKD with an MCMI-trained teacher reaches 49.53\% Top-1 accuracy and 97.65\% Top-5 accuracy, which improves the Top-1 accuracy by 2.41 percentage points over the offline baseline. In contrast, the 4-bit HintKD variant with an MLL-trained teacher achieves 46.95\% Top-1 accuracy, which is lower than the MCMI-based counterpart. This shows that, for Scenario 9, introducing MCMI into teacher training can provide more informative guidance for the student and thus lead to better final performance.

We further analyze the effect of hint bandwidth on Scenario 9. When the hint bandwidth is reduced from 4 bits to 2 bits, the Top-1 accuracy decreases to 45.44\%, indicating that an overly small codebook cannot preserve sufficient discriminative information. When the bandwidth is increased from 4 bits to 6 bits, the Top-1 accuracy becomes 47.07\%, which still does not surpass the 4-bit setting. Therefore, the 4-bit configuration provides the best trade-off between compact representation and predictive performance on Scenario 9.

In addition to predictive accuracy, HintKD also exhibits clear communication advantages. On Scenario 9, Feature KD and Logits KD require 2688 bits and 1056 bits per sample during inference, respectively, whereas HintKD only needs 4 bits in the default setting. This corresponds to compression ratios of 672× and 264× relative to Feature KD and Logits KD, respectively. Even during training, the transmission cost of HintKD remains lower than that of Feature KD. These results show that the proposed method can significantly reduce communication overhead while preserving competitive beam prediction performance, which is particularly suitable for bandwidth-constrained cloud-edge inference scenarios. The communication overhead comparison is also illustrated in Fig. 3.

To further verify the generality of the proposed distillation framework, we additionally conduct experiments on the MNIST dataset. Compared with the radar-based beam prediction task in Scenario 9, MNIST is a conventional handwritten digit recognition task, which differs significantly in data form, semantic structure, and application background. Therefore, the MNIST results are used to demonstrate that HintKD is not only effective for wireless sensing aided beam prediction, but also applicable to standard image classification tasks. As shown in Table~\ref{tab:main_results}, the offline student on MNIST achieves 99.05\% Top-1 accuracy. Under the 4-bit setting, HintKD with an MLL-trained teacher reaches 99.17\% Top-1 accuracy, while the MCMI-based HintKD reaches 99.06\%. At the same time, the Top-5 accuracy remains stable at 99.97\%, indicating that the proposed method can maintain highly competitive classification performance even with an extremely low-bandwidth hint transmission mechanism.
\begin{figure}[t]
     \centering     
     \includegraphics[width=1\linewidth]{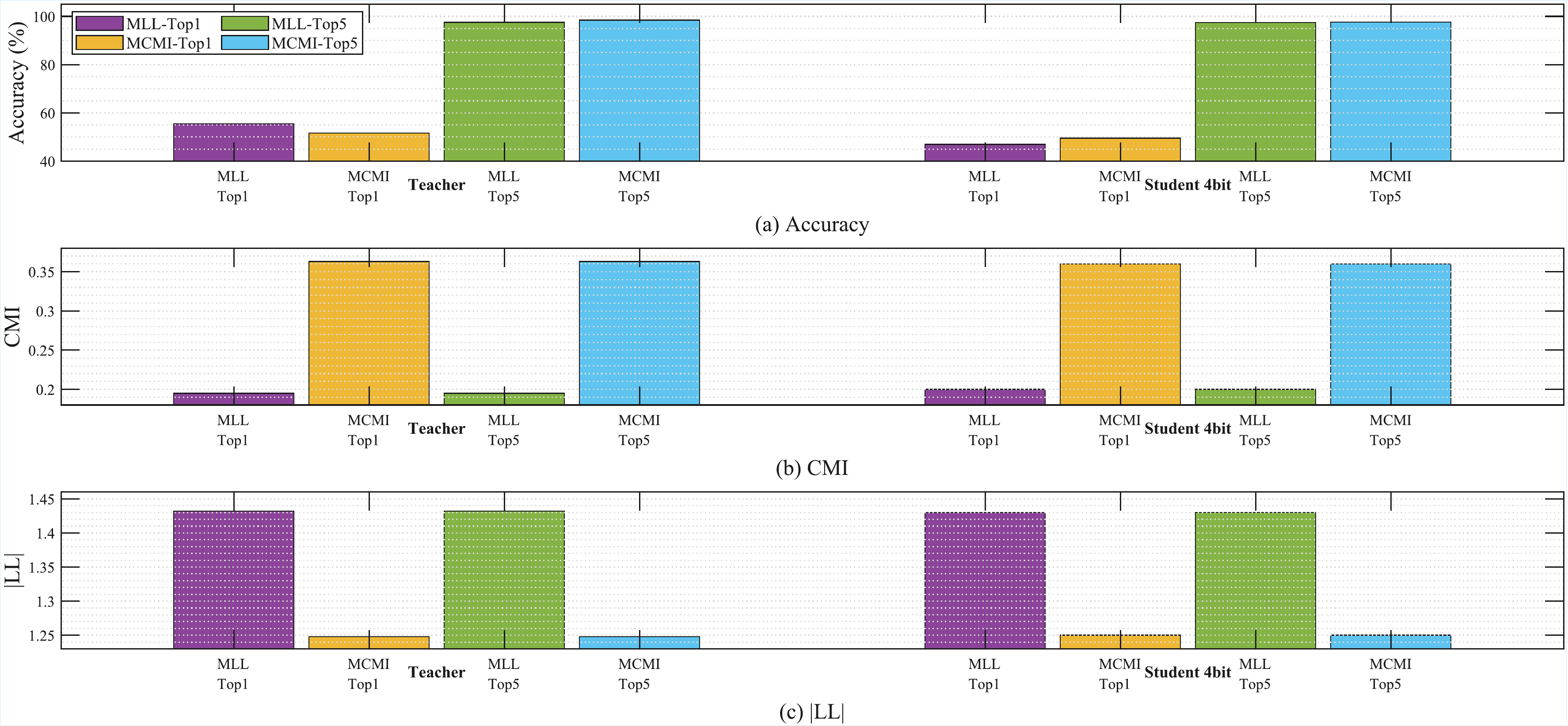}
     \caption{The effect of methods on accuracy, CMI and LL.}
     \label{fig:mll}
\end{figure}
\begin{figure}[t]
    \centering
    \includegraphics[width=1\linewidth]{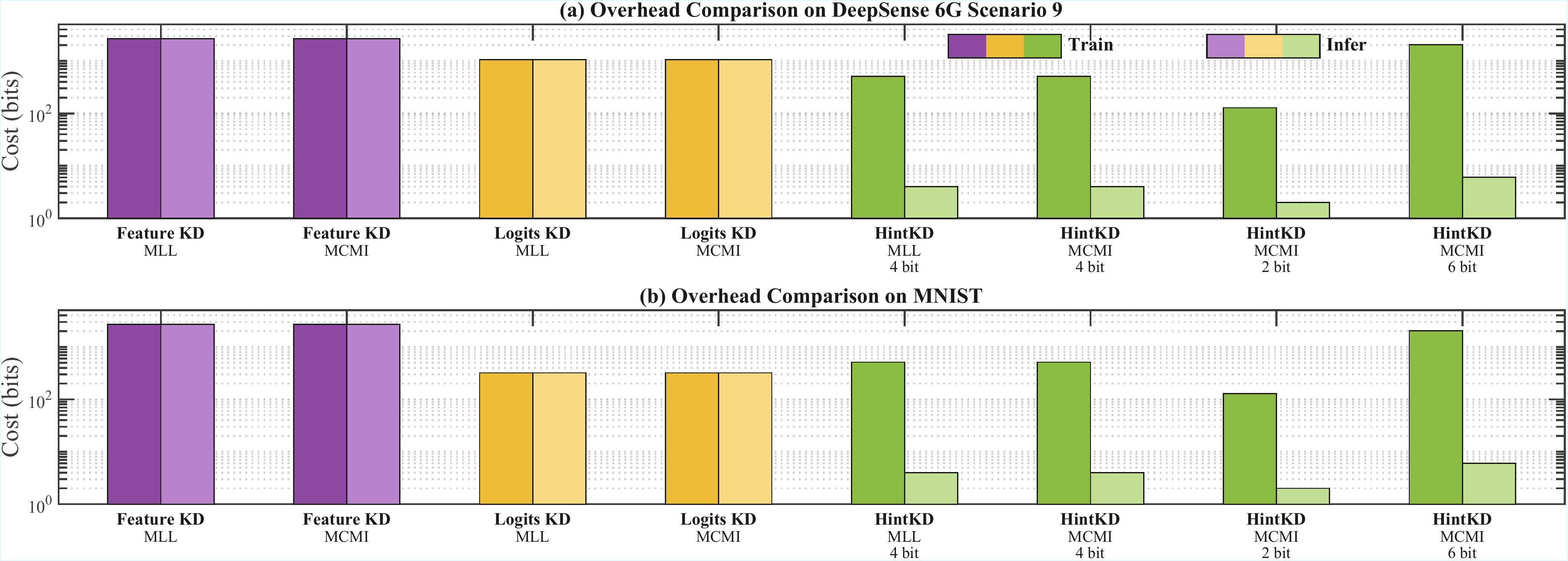} % 替换为你的 Top-1 图片文件名
    \caption{Overhead comparison on Deepsense 6G Scenario9 and MNIST.}
    \label{fig:cost}
\end{figure}
From the perspective of communication overhead, the advantage of HintKD is still evident on MNIST. In the 4-bit setting, its training-side cost is 512 bits and inference-side cost is only 4 bits, which is far lower than Feature KD and also much lower than the inference overhead of Logits KD. Combined with the results on Scenario 9, these experiments show that HintKD consistently achieves a favorable accuracy-bandwidth trade-off across different datasets and tasks. Therefore, the MNIST experiment serves as an important supplementary validation, confirming the universality and cross-task applicability of the proposed distillation method rather than limiting the conclusion to a single beam prediction scenario.

Overall, the results of Fig.~\ref{fig:mll}, Table~\ref{tab:main_results}, and Fig.~\ref{fig:cost} consistently demonstrate that the proposed HintKD framework can effectively compress teacher guidance into low-bit discrete hints and transmit them efficiently in cloud-edge inference. On the DeepSense 6G Scenario 9 task, the proposed method achieves a strong balance between predictive accuracy and bandwidth cost, while on MNIST it further verifies the generality of the framework on standard image classification benchmarks. These results confirm that HintKD is a promising solution for bandwidth-constrained collaborative intelligence systems.
\vspace{-0.2cm}
\section{Conclusion}
\vspace{-0.1cm}
This paper proposed HintKD, a cloud-edge distillation framework that replaces high-dimensional teacher transmissions with compact discrete hints. By combining an MCMI-trained teacher, differentiable codebook learning, and FiLM-based hint adaptation on the UE, the proposed method converts remote teacher guidance into a low-bit control signal. Experiments on DeepSense 6G 
and MNIST 
show that HintKD preserves competitive accuracy while reducing the inference-side payload to only a few bits per sample. These results suggest that, for bandwidth-constrained edge intelligence, transmitting structured guidance can be considerably more efficient than transmitting raw features or full logit vectors. Future work will investigate adaptive codebooks for time-varying channels and extend the framework to broader multimodal sensing scenarios.
\bibliographystyle{IEEEtran}
\bibliography{main}

\end{document}